\documentclass[12pt]{article}
\usepackage{geometry}
\usepackage{cite}

\usepackage{paralist}
\usepackage{graphicx}
\usepackage{subfigure}
\usepackage{url}

\usepackage{color}

\usepackage{amsmath}
\usepackage{amssymb}

\allowdisplaybreaks[1]

\usepackage{amsthm}

\newtheorem{theorem}{Theorem}
\newtheorem{lemma}[theorem]{Lemma}

\makeatletter
\def\@endtheorem{\endtrivlist}
\makeatother

\newcommand{\substr}[3]{#1[#2..#3]}

\newcommand{\scheme}{\mathcal{S}}
\newcommand{\schemelam}{\mathcal{S}_{\text{Lam}}}

\newcommand{\partition}{X}
\newcommand{\numstrings}{\mathit{strings}}
\newcommand{\numstringsp}{\numstrings'}
\newcommand{\strings}{\mathcal{A}}
\newcommand{\numnodes}[1]{\mathit{nodes}_{#1}}
\newcommand{\numnodesd}[2]{\mathit{nodes}_{#1,#2}}
\newcommand{\prob}[2]{\mathit{prob}_{#1,#2,\sigma}}
\newcommand{\Psearch}{P^*}

\newcommand{\numstringsedit}{\numstrings_{\mathrm{edit}}}
\newcommand{\numnodesp}[2]{\mathit{nodes}_{#1}(#2)}
\newcommand{\numnodesqp}[3]{\mathit{nodes}_{#1,#2}(#3)}
\newcommand{\numnodesqpb}[3]{\mathit{nodes}_{#1,#2,#3}}

\newcommand{\automaton}[1]{\mathcal{A}_{#1}}
\newcommand{\transition}[1]{\hat{\delta}_{#1}}
\newcommand{\states}[1]{\mathrm{Q}_{#1}}

\newcommand{\numstringspp}{\numstrings''}
\newcommand{\logn}{N}

\newcommand{\lognc}{N}
\newcommand{\ppartition}{Y}
\newcommand{\spartition}{Z}

\newcommand{\plen}[1]{m_{#1}}
\newcommand{\plast}[1]{l_{#1}}
\newcommand{\pnumparts}[1]{p_{#1}}
\newcommand{\prank}[1]{r_{#1}}
\newcommand{\pprefix}[1]{\mathit{prefix}(#1)}
\newcommand{\optv}[3]{v(#1,#2,#3)}

\newcommand{\seq}[2]{[#1,#2]}
\newcommand{\seqr}[2]{\overline{[#1,#2]}}

\newif\iffull
\fulltrue
\newif\ifextra
\extrafalse

\begin{document}

\title{Approximate String Matching using a Bidirectional Index}
\author{
Gregory Kucherov%
\thanks{CNRS/LIGM, Universit\'e Paris-Est Marne-la-Vall\'ee, France}
\and
Kamil Salikhov$^*$\thanks{Mechanics and Mathematics Department, Lomonosov Moscow State University, Russia}
\and
Dekel Tsur%
\thanks{Department of Computer Science, Ben-Gurion University of
the Negev, Israel}
}
\date{}
\maketitle

\begin{abstract}
We study strategies of approximate pattern matching that exploit
bidirectional text indexes, extending and generalizing ideas of
\cite{LamLTWWY09}. We introduce a formalism, called search schemes, to specify search strategies of
this type, then develop a probabilistic measure for the efficiency of
a search scheme, prove several combinatorial results on efficient
search schemes, and finally, provide experimental computations
supporting the superiority of our strategies. 
\end{abstract}

\section{Introduction}

Approximate
string matching has numerous practical applications and has long been a subject of extensive studies by 
algorithmic researchers~\cite{Navarro:2001:GTA:375360.375365}. If errors are allowed in a
match between a pattern string and a text string, most of
fundamental ideas behind exact string search algorithms become
inapplicable.

The problem of approximate string matching comes in
different variants. In this paper, we are concerned with the
\emph{indexed} variant, when a static text is available for
pre-processing and storing in a data structure (index), before any matching
query is made. 
The 
challenge of indexed approximate matching is to construct a small-size
index supporting quick search for approximate 
pattern occurrences, within a worst-case time weakly dependent on the text
length. 
From the theoretical perspective, 
even the case of one allowed error turned out to be
highly nontrivial and gave rise to a series of works (see
\cite{LamSW05} and references therein). In the case of $k$ errors,
existing solutions generally have time or space complexity that is exponential
in $k$, see~\cite{SungEncyclopedia08} for a survey.

The quest for efficient approximate string matching algorithms has
been boosted by a new generation of DNA sequencing
technologies, capable to produce huge quantities of short 
DNA sequences, called \emph{reads}. 
Then, an important task is to \emph{map} 
those reads to a given reference genomic sequence,
which requires very fast and accurate approximate string matching
algorithms. 
This motivation resulted in a very large number of read mapping algorithms and
associated software programs, 
we refer to~\cite{Li01092010} for a survey.

Broadly speaking, 
read mapping algorithms follow one of two main approaches, or sometimes
a combination of those. 
The \emph{filtration} approach proceeds in two
steps: it first identifies (with or without using a full-text index)
locations of the text where the pattern can \emph{potentially} occur,
and then verifies these locations for actual matches. Different filtration schemes have been proposed
\cite{NavarroRaffinot-book02,FarachLST-conf,KucherovNoeRoytbergJCBB05,KarkkainenN07}. Filtration
algorithms usually don't offer interesting worst-case time and space bounds but are
often efficient on average and are widely used in practice. 
Another approach, usually called \emph{backtracking}, 
extends exact matching algorithms to the approximate case by some
enumeration of possible errors and by simulating exact search of all
possible variants of the pattern. It is this approach that we follow in
the present work. Backtracking and filtration techniques can be
combined in a \emph{hybrid} approach~\cite{NavarroB00}. 

Some approximate matching algorithms use standard text indexes, such as
suffix tree or suffix arrays. However, for large
datasets occurring in modern applications, these indexes are
known to take too much memory. Suffix arrays and suffix trees
typically require at least 4 or 10 \emph{bytes} per character respectively. 
The last years saw the development of \emph{succinct} or \emph{compressed full-text indexes}
that occupy virtually as much memory as the sequence itself and yet
provide very powerful functionalities~\cite{NavarroM07}.
For example, the FM-index~\cite{FerraginaM00},
based on the Burrows-Wheeler Transform~\cite{BurrowW94}, may
occupy 2--4 \emph{bits} of memory {per character} for DNA texts. FM-index has now been used in
many practical bioinformatics software programs,
e.g.~\cite{LangmeadTPS09,LiD09,Simpson01032012}.
Even if succinct indexes are primarily
designed for exact string search, using them for approximate matching
naturally became an attractive opportunity. This direction has been
taken in several papers, see~\cite{RussoEtAlAlgorithms09}, as well as
in practical implementations~\cite{Simpson01032012}. 

Interestingly, succinct indexes can provide even more functionalities
than classical ones. In particular, succinct indexes can be
made \emph{bidirectional}, i.e.\ can perform pattern search in both
directions 
\cite{LamLTWWY09,RussoEtAlAlgorithms09,SchnattingerOG12,BelazzouguiCKM13}.
Lam et al.~\cite{LamLTWWY09} showed how a
bidirectional FM-index can be used to efficiently search for strings
up to a small number (one or two) errors. The idea is to partition the pattern into $k+1$
equal parts, where $k$ is the number of errors,
and then perform multiple searches on the FM-index, where
each search assumes a different distribution of mismatches among the
pattern parts. It has been shown experimentally
in~\cite{LamLTWWY09} that this improvement leads to a
faster search compared to the best existing read alignment
software.
Related algorithmic ideas appear also in~\cite{RussoEtAlAlgorithms09}.

In this paper, we extend the search strategy of~\cite{LamLTWWY09} 
in two
main directions. We consider the case of arbitrary $k$ and propose to
partition the pattern into more than $k+1$ parts that can be of
\emph{unequal} size. To demonstrate the benefit of both ideas, we first introduce a general formal framework for this
kind of algorithm, called \emph{search scheme}, that allows us to easily specify them and to
reason about them (Section~\ref{sec:bidirectional}). 
Then, in Section~\ref{sec:analysis} we perform a probabilistic analysis
that provides us with a quantitative measure of performance of a search
scheme, and give an efficient algorithm for obtaining the optimal
pattern partition for a given scheme.
Furthermore, we prove
several combinatorial results on the design of efficient search
schemes (Section~\ref{sec:design}).
Finally, Section~\ref{sec:experiments} contains comparative analytical
estimations, based on our probabilistic analysis, that demonstrate the
superiority of our search strategies for many practical parameter
ranges. 
We further report on large-scale experiments on genomic data
supporting this analysis. 

\section{Bidirectional search}\label{sec:bidirectional}

In the framework of text indexing, pattern search is usually done by
scanning the pattern online and recomputing \emph{index points}
referring to the 
occurrences of the scanned part of the pattern. 
With classical text indexes, such as suffix trees or
suffix arrays, the pattern is scanned left-to-right (\emph{forward
search}). However, some compact indexes such as FM-index provide a
search algorithm that scans the pattern right-to-left (\emph{backward
search}). 

Consider now approximate string matching. For ease of presentation, we
present most of our ideas for the
case of Hamming distance
(recall that the Hamming distance between two strings $A$ and $B$ of equal
lengths is the number of indices $i$ for which $A[i] \neq B[i]$),
although our algorithms extend to the edit
distance as well. Section~\ref{sec:estimation-edit} below will
specifically deal with the edit distance.

Assume that $k$ letter mismatches are allowed between a pattern $P$ and
a substring of length $|P|$ of a text $T$.
Both forward and backward search can be extended to
approximate search in a straightforward way, by exploring all
possible mismatches along the search, as long as their number does
not exceed $k$ and the current pattern still occurs in the text.
For the forward search, for example, the algorithm enumerates all
substrings of $T$ with
Hamming distance at most $k$ to a \emph{prefix} of $P$.
Starting with the empty string, the enumeration 
is done by extending the current string with the corresponding letter
of $P$, and with all other letters provided that the number of accumulated
mismatches has not yet reached $k$. For each extension,
its positions in $T$ are computed using the index.
Note that the set of enumerated strings is closed under prefixes and
therefore can be represented by the nodes of a trie.
Similar to forward search, \emph{backward search}
enumerates all substrings of $T$
with Hamming distance at most $k$ to a \emph{suffix} of $P$.

Clearly, backward and forward search are symmetric and, once we have
an implementation of one, the other can be implemented similarly by constructing the index
for the reversed text. 
However, combining both forward and backward search within one algorithm
results in a more efficient search.
To illustrate this, consider the case $k=1$.
Partition $P$ into two equal length parts $P=P_1 P_2$.
The idea is to perform two complementary searches: forward search for
occurrences of $P$ with a mismatch in $P_2$ and backward search for
occurrences with a mismatch in $P_1$. In both searches, branching is
performed only after $|P|/2$ characters are matched. 
Then,
the number of strings enumerated by the
two searches is much less than the number of strings enumerated by a single
standard forward search, even though two searches are performed instead
of one.

A \emph{bidirectional index} of a text allows one to extend the current
string $A$ both left and right, that is, compute the positions of either
$cA$ or $Ac$ from the positions of $A$.
Note that a bidirectional
index allows forward and backward searches to alternate, which will be
crucial for our purposes. 
Lam et al.~\cite{LamLTWWY09} showed how the FM-index can be made
bidirectional. Other succinct bidirectional indexes were given
in~\cite{RussoEtAlAlgorithms09,SchnattingerOG12,BelazzouguiCKM13}.
Using a bidirectional index, such as FM-index, forward and backward searches can be
performed in time linear in the number of enumerated strings.
Therefore, our main goal is to organize the search so that the number
of enumerated strings is minimized.

Lam et al.~\cite{LamLTWWY09} gave a new search algorithm, called
\emph{bidirectional search}, that utilizes the bidirectional property
of the index.
Consider the case $k=2$, studied in~\cite{LamLTWWY09}.
In this case, the pattern is partitioned into three equal length parts,
$P=P_1 P_2 P_3$.
There are now 6 cases to consider according to the placement of mismatches
within the parts:
011 (i.e.\ one mismatch in $P_2$ and one mismatch in $P_3$),
101, 110, 002, 020, and 200.
The algorithm of Lam et al.~\cite{LamLTWWY09} performs three searches
(illustrated in Figure~\ref{fig:tries}):
\begin{enumerate}
\item A forward search that allows no mismatches when processing
characters of $P_1$, and 0 to 2 accumulated mismatches when processing 
characters of $P_2$ and $P_3$.
This search handles the cases 011, 002, and 020 above.
\item
A backward search that allows no mismatches when processing characters
of $P_3$,
0 to 1 accumulated mismatches when processing characters of $P_2$, and
0 to 2 accumulated mismatches when processing characters of $P_1$.
This search handles the cases 110 and 200 above.
\item
The remaining case is 101.
This case is handled using a \emph{bidirectional search}.
It starts with a forward search on string $P'=P_2 P_3$ that
allows no mismatches when processing characters of $P_2$, and
0 to 1 accumulated mismatches when processing the characters of $P_3$.
For each string $A$ of length $|P'|$ enumerated by the forward search whose
Hamming distance from $P'$ is exactly 1, a backward search for $P_1$ is performed
by extending $A$ to the left, 
allowing one additional mismatch.
In other words, the search allows 1 to 2 accumulated mismatches when processing the
characters of $P_1$.
\end{enumerate}
\begin{figure}
\centering
\subfigure[Forward search]{\includegraphics[scale=0.5]{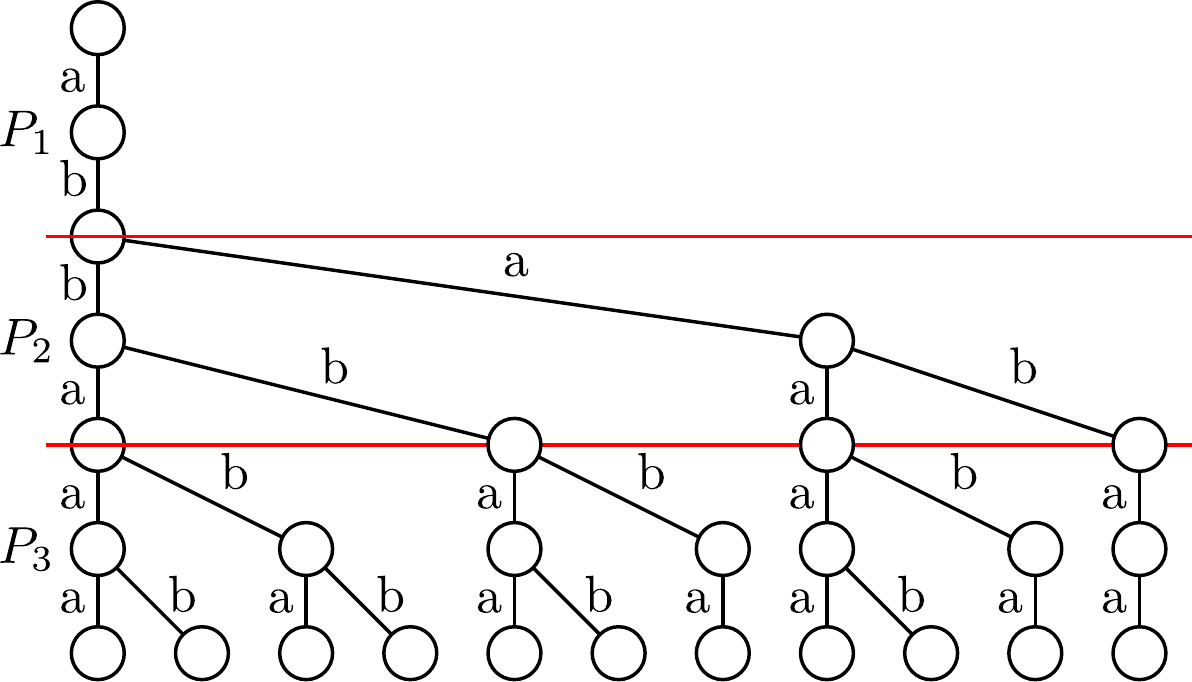}}
\subfigure[Backward search]{\includegraphics[scale=0.5]{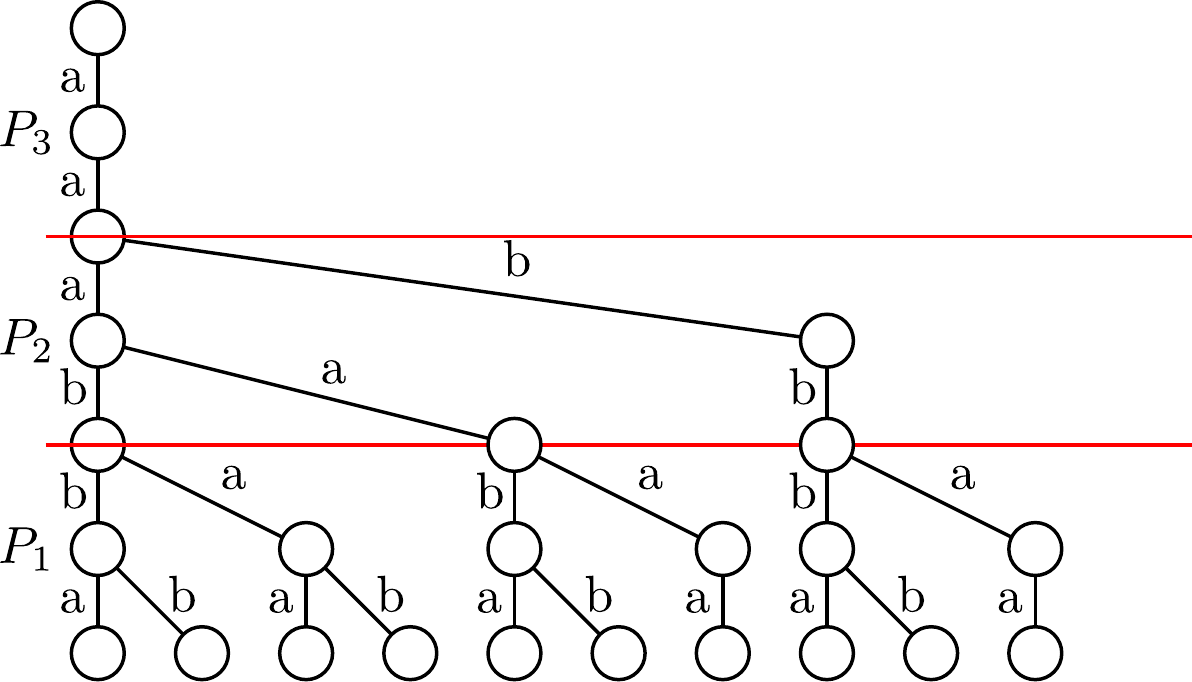}}
\subfigure[Bidirectional search\label{fig:trie-Sbd}]
 {\includegraphics[scale=0.5]{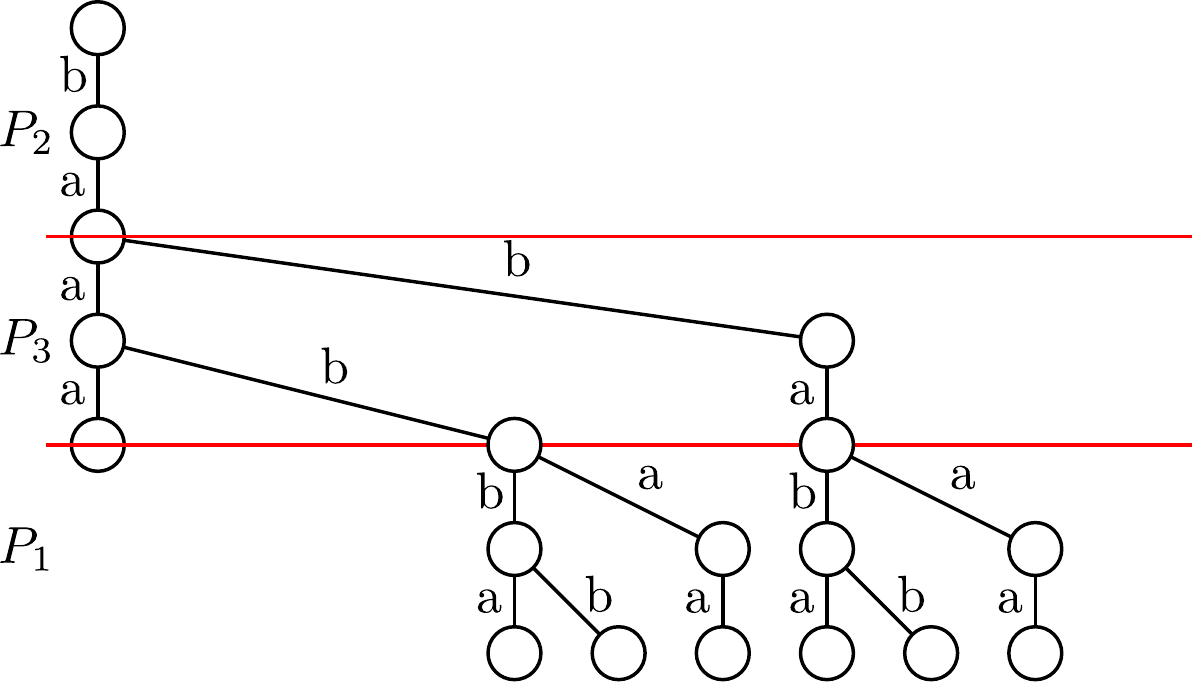}}
\caption{The tries representing the searches of Lam et al.\ for binary
alphabet $\{\mathrm{a},\mathrm{b}\}$, search string $P=\mathrm{abbaaa}$,
and number of errors $2$.
Each trie represents one search 
and assumes that all the enumerated substrings exist in the text $T$.
In an actual search on a specific $T$, each trie contains of a subset of
the nodes, depending on whether the strings of the nodes in the trie appear
in $T$.
A vertical edge represents a match, and a diagonal edge represents a mismatch.
\label{fig:tries}}
\end{figure}

We now give a formal definition for the above.
Suppose that the pattern $P$ is partitioned into $p$ parts.
A \emph{search} is a triplet of strings $S=(\pi,L,U)$
where $\pi$ is a permutation string of length $p$ over $\{1,\ldots,p\}$,
and $L,U$ are strings of length $p$ over $\{0,\ldots,k\}$.
The string $\pi$ indicates the order in which the parts of $P$ are processed,
and thus it must satisfy the following \emph{connectivity property}:
For every $i>1$,
$\pi(i)$ is either $(\min_{j<i}\pi(j))-1$ or $(\max_{j<i}\pi(j))+1$.
The strings $U$ and $L$ give upper and lower bounds on the number of
mismatches:
When the $j$-th part is processed, the number of accumulated mismatches
between the active strings and the corresponding substring of $P$ must be between
$L[j]$ and $U[j]$.
Formally, for a string $A$ over integers,
the \emph{weight} of $A$ is $\sum_i A[i]$.
A search $S=(\pi,L,U)$ \emph{covers} a string $A$ if
$L[i+1] \leq \sum_{j=1}^i A[j] \leq U[i]$ for all $i$
(assuming $L[p+1]=0$).
A \emph{$k$-mismatch search scheme} $\scheme$ is a collection of
searches such that for every string $A$ of weight $k$,
there is a search in $\scheme$ that covers $A$.
For example, the 2-mismatch scheme of Lam et al.\ 
consists of
searches $S_f = (123,000,022)$, $S_b = (321,000,012)$,
and $S_{bd} = (231,001,012)$. We denote this scheme by $\schemelam$.

In this work, we introduce two types of improvements over the search scheme of Lam et al.
\paragraph{Uneven partition.}
In $\schemelam$, search $S_f$ enumerates more strings
than the other two searches, as it allows 2 mismatches on the second
processed part
of $P$, while the other two searches allow only one mismatch.
If we increase the length of $P_1$ in the partition of $P$, the
number of strings enumerated by $S_f$ will decrease, while
the number of strings enumerated by the two other searches will increase.
We show that for some typical parameters of the problem, the decrease in the
former number is larger than the increase of the latter number,
leading to a more efficient search. 

\paragraph{More parts.}
Another improvement can be achieved using partitions with $k+2$ or
more parts, rather than $k+1$ parts.
\iffull
We explain in Section~\ref{sec:uneven} why such partitions can
reduce the number of enumerated strings.
\fi

\section{Analysis of search schemes}\label{sec:analysis}
In this section we show how to estimate the performance of a given search scheme
$\scheme$. Using this technique, we \iffull first explain why an
uneven partition can lead to a better performance, and then \fi present a
dynamic programming algorithm for designing an
optimal partition of a pattern. 

\subsection{Estimating the efficiency of a search scheme}\label{sec:estimation}
To measure the efficiency of a search scheme, we estimate the number of strings
enumerated by all the searches of $\scheme$.
We assume that performing single steps of
forward, backward, or bidirectional searches takes the same amount of time.
It is fairly straightforward to extend the method of this section to the
case when these times are not equal.
Note that the bidirectional index of Lam et al.~\cite{LamLTWWY09}
reportedly spends slightly more time (order of 10\%) on forward search
than on backward search. 

For the analysis, we assume that characters
of $T$ and $P$ are randomly drawn uniformly and independently
from the alphabet.
We note that it is possible to extend the method of this section to a
non-uniform distribution.
For more complex distributions,
a Monte Carlo simulation can be applied
which, however, 
requires much more time than
the method of this section.

\subsubsection{Hamming distance}\label{sec:hamming-estimation}

Our approach to the analysis is as follows.
Consider a fixed search $S$, and the trie representing this search 
(see Figure~\ref{fig:tries}).
The search enumerates the largest number of strings when
the text contains all strings of length $m$ as substrings.
In this case, every string that occurs in the trie is enumerated.
For other texts, the set of enumerated strings is a subset of the set of strings
that occurs in trie.
The expected number of strings enumerated by $S$ on random $T$ and $P$ is equal
to the sum over all nodes $v$ of the trie of the probability
that the corresponding string appears in $T$.
We will first show that this probability depends only on the depth of $v$
(Lemmas~\ref{lem:substring-probability} and~\ref{lem:Ali} below).
Then, we will show how to count the number of nodes in each level of the
trie.

Let $\prob{n}{l}$ denote the probability that a random string of length $l$
is a substring of a random string of length $n$, where the characters
of both strings are randomly chosen uniformly and independently
from an alphabet of size $\sigma$.
The following lemma gives an approximation for $\prob{n}{l}$ with
a bound on the approximation error.
\begin{lemma}\label{lem:substring-probability}
$|\prob{n}{l} - (1-e^{-n/\sigma^l})| \leq
\begin{cases}
 4nl/\sigma^{2l} & \text{if }l \geq \log_{\sigma} n \\
4l/\sigma^l & \text{otherwise}
\end{cases}$.
\end{lemma}
\begin{proof}
Let $A$ and $B$ be random strings of length $l$ and $n$, respectively.
Let $E_i$ be the event that $A$ appears in $B$ at position $i$.
The event $E_i$ is independent of the events
$\{E_j: j \in \{1,2,\ldots,n-l+1\} \setminus F_i\}$, where
$F_i =\{i-l+1,i-l+2,\ldots,i+l-1 \}$.
By the Chen-Stein method~\cite{Chen75,Barbour92},
\[
\left|\prob{n}{l} - (1-e^{-n/\sigma^l})\right| \leq
\frac{1-e^{-\lambda}}{\lambda}
\sum_{i=1}^{n-l+1} \sum_{j\in F_i} (\Pr[E_i]\Pr[E_j]+\Pr[E_i \cap E_j]),
\]
where $\lambda = n/\sigma^l$.
Clearly, $\Pr[E_i] = \Pr[E_j] = 1/\sigma^l$.
It is also easy to verify that $\Pr[E_i \cap E_j] = 1/\sigma^{2l}$.
Therefore,
$|\prob{n}{l} - (1-e^{-n/\sigma^l})| \leq
((1-e^{-\lambda})/\lambda) \cdot 4nl/\sigma^{2l}$.
The lemma follows since $(1-e^{-\lambda})/\lambda \leq \min(1,1/\lambda)$ 
for all $\lambda$.
\end{proof}
The bound in Lemma~\ref{lem:substring-probability} on the error of the
approximation of $\prob{n}{l}$ is large if $l$ is small,
say $l < \frac{1}{2}\log_{\sigma} n$.
In this case, we can get a better bound by observing that
$\prob{n}{l} \geq \prob{n}{l_0}$, where $l_0 = \frac{3}{4} \log_{\sigma} n$.
Since $\prob{n}{l_0} \geq 1-e^{-n/\sigma^{l_0}} - 4l_0/\sigma^{l_0}$,
we obtain that $|\prob{n}{l} - (1-e^{-n/\sigma^l}) | \leq
\max(e^{-n/\sigma^l},e^{-n/\sigma^{l_0}} + 4 l_0/\sigma^{l_0})$.

Let $\numstrings(S,\partition,\sigma,n)$ denote the expected number of strings
enumerated when performing a search $S=(\pi,L,U)$ on a random
text of length $n$ and random pattern of length $m$, where
$\partition$ is a partition of the pattern and
$\sigma$ is the alphabet size 
(note that $m$ is not a parameter for $\numstrings$ since the value
of $m$ is implied from $\partition$).
For a search scheme $\scheme$, $\numstrings(\scheme,\partition,\sigma,n) =
\sum_{S\in \scheme} \numstrings(S,\partition,\sigma,n)$.

Fix $S$, $\partition$, $\sigma$, and $n$.
Let $\strings_l$ be the set of enumerated strings of length $l$
when performing search $S$ on a random pattern of length $m$,
partitioned by $\partition$,
and a text $\hat{T}$ containing all strings of length at most $m$ as
substrings.
Let $A_{l,i}$ be the $i$-th element of $\strings_l$
(an order on $\strings_l$ will be defined in the proof of the next lemma).
Let $\numnodes{l} = |\strings_l|$, namely, the number of nodes
at depth $l$ in the trie that represents the search $S$.
Let $\Psearch$ be the string containing the characters of $P$ according to the
order they are read by the search.
In other words, $\Psearch[l]$ is the character such that 
every node at depth $l-1$ of the trie has an edge to a child
with label $\Psearch[l]$.

\begin{lemma}\label{lem:Ali}
For every $l$ and $i$, the string $A_{l,i}$ is a random string with uniform
distribution.
\end{lemma}
\begin{proof}
Assume that the alphabet is $\Sigma=\{0,\ldots,\sigma-1\}$.
Consider the trie that represents the search $S$.
We define an order on the children of each node of the trie as follows:
Let $v$ be a node in the trie with depth $l-1$.
The label on the edge between $v$ and its leftmost child is 
$\Psearch[l]$.
If $v$ has more than one child, the labels on the edges to the rest of the
children of $v$, from left to right,
are $(\Psearch[l]+1)\bmod \sigma,\ldots,(\Psearch[l]+\sigma-1)\bmod \sigma$.
We now order the set $\strings_l$ according to the nodes of depth
$l$ in the trie. Namely, let $v_1,\ldots,v_{\numnodes{l}}$
be the nodes of depth $l$ in the trie, from left to right.
Then, $A_{l,i}$ is the string that corresponds to $v_i$.
We have that $A_{l,i}[j]=(\Psearch[j]+c_{i,j}-1)\bmod \sigma$ for
$j = 1,\ldots,l$,
where $c_{i,j}$ is the rank of the node of depth $j$ on the path from the
root to $v_i$ among its siblings.
Now, since each letter of $P$ is randomly chosen uniformly and independently
from the alphabet, it follows that each letter of $A_{l,i}$ has uniform
distribution and the letters of $A_{l,i}$ are independent.
\end{proof}

By the linearity of the expectation,
\[ \numstrings(S,\partition,\sigma,n) =
 \sum_{l \geq 1} \sum_{i=1}^{\numnodes{l}}
 \Pr_{T\in \Sigma^n}[\text{$A_{l,i}$ is a substring of $T$}].
\]
By Lemma~\ref{lem:Ali} and Lemma~\ref{lem:substring-probability},
\begin{equation}\label{eq:a}
\numstrings(S,\partition,\sigma,n) = 
 \sum_{l=1}^{m} \numnodes{l} \cdot \prob{n}{l} \approx
 \sum_{l = 1}^{m} \numnodes{l} (1-e^{-n/\sigma^{l}}).
\end{equation}
We note that the bounds on the approximation errors of $\prob{n}{l}$
are small, therefore even when these bounds are multiplied by $\numnodes{l}$
and summed over all $l$, the resulting bound on the error is small.

In order to compute the values of $\numnodes{l}$, we give some definitions.
Let $\numnodesd{l}{d}$ be the number of strings in $\strings_l$
of length $l$ with Hamming distance $d$ to the prefix of $\Psearch$
of length $l$.
For example, consider search $S_{bd} = (231,001,012)$ and partition of 
a pattern of length 6 into 3 parts of length 2,
as shown in Figure~\ref{fig:trie-Sbd}.
Then, $\Psearch = \mathrm{baaaba}$, $\numnodesd{5}{0} = 0$,
$\numnodesd{5}{1} = 2$ (strings baabb and babab),
and $\numnodesd{5}{2} = 2$ (strings baaba and babaa).

Let $\pi_{\partition}$ be a string obtained from $\pi$ by replacing each
character $\pi(i)$ of $\pi$ by a run of $\pi(i)$ of length
$\partition[{\pi(i)}]$,
where $\partition[j]$ is the length of the $j$-th part in the partition
$\partition$.
Similarly, $L_{\partition}$ is a string obtained from
$L$ by replacing each character $L[i]$  by a run of $L[i]$ of length
$\partition[{\pi(i)}]$, and $U_{\partition}$ is defined analogously.
In other words, values $L_{\partition}[i],U_{\partition}[i]$
give lower and upper bounds on the number of allowed mismatches for an
enumerated string of length $i$.
For example, for $S_{bd}$ and the partition $\partition$ defined above,
$\pi_{\partition} = 223311$, $L_{\partition} = 000011$,
and $U_{\partition} = 001122$.

Values $\numnodes{l}$ are given by the following recurrence.
\begin{align}\label{eqn:nl}
\numnodes{l}  & = 
 \sum_{d = L_{\partition}[l]}^{U_{\partition}[l]} \numnodesd{l}{d}\\
\numnodesd{l}{d} & = \begin{cases}
\numnodesd{l-1}{d} + (\sigma-1) \cdot\numnodesd{l-1}{d-1} & 
\text{if $l \geq 1$ and $L_{\partition}[l] \leq d \leq U_{\partition}[l]$} \\
1 & \text{if $l = 0$ and $d = 0$}\\
0 & \text{otherwise}
\end{cases}
\end{align}
For a specific search, a closed formula can be given for $\numnodes{l}$.
If a search scheme $\scheme$ contains two or more searches with the same
$\pi$-strings, these searches can be merged in order to eliminate the
enumeration of the same string twice or more.
It is straightforward to modify the computation of $\numstrings(\scheme,\partition,\sigma,n)$ to account for this optimization.

Consider equation~(\ref{eq:a}).
The value of the term $1-e^{-n/\sigma^{l}}$ is very close to 1
for $l \leq \log_{\sigma}n - O(1)$.
When $l \geq \log_{\sigma}n$, the value of this term decreases
exponentially.
Note that $\numnodes{l}$ increases exponentially, but the base of the exponent
of $\numnodes{l}$ is $\sigma-1$ whereas the base of $1-e^{-n/\sigma^{l}}$ is
$1/\sigma$.
We can then approximate $\numstrings(S,\partition,\sigma,n)$ with function
$\numstringsp(S,\partition,\sigma,n)$ defined by
\begin{equation}\label{eq:a2}
\numstringsp(S,\partition,\sigma,n) =
\sum_{l = 1}^{\lceil\log_{\sigma} n\rceil+c_\sigma} \numnodes{l}\cdot
 (1-e^{-n/\sigma^{l}}),
\end{equation}
where $c_\sigma$ is a constant chosen so that $((\sigma-1)/\sigma)^{c_\sigma}$
is sufficiently small.

From the above formulas we have that the time complexities for
computing $\numstrings(\scheme,\partition,\sigma,n)$ and
$\numstringsp(\scheme,\partition,\sigma,n)$ are
$O(|\scheme|km)$ and $O(|\scheme|k \log_{\sigma} n)$, respectively.

\subsubsection{Edit distance}\label{sec:estimation-edit}
We now show how to estimate the efficiency of a search scheme for the edit distance.

We define $\numstringsedit$ analogously to $\numstrings$
in the previous section, except that edit distance errors are allowed.
Fix a search $S=(\pi,L,U)$ and a partition $\partition$.
We assume without loss of generality that $\pi$ is the identity permutation.
Similarly to the Hamming distance case,
define $\strings_l$ to be the set of enumerated strings of length $l$
when performing the search $S$ on a random pattern of length $m$,
partitioned by $\partition$,
and a text $\hat{T}$ containing all the strings of length at most $m+k$ as
substrings.
Unlike the case of Hamming distance, here the strings of $\strings_l$ are not
distributed uniformly.
Thus, we do not have the equality
$\numstringsedit(S,\partition,\sigma,n) = \sum_{l=1}^{m} \numnodes{l}
\cdot \prob{n}{l}$.
We will use $\sum_{l=1}^{m} \numnodes{l} \cdot \prob{n}{l}$ as an approximation
for $\numstringsedit(S,\partition,\sigma,n)$, but we do not have an estimation
on the error of this approximation.
Note that in the Hamming distance case, the sizes of the sets $\strings_l$ are
the same for every choice of the pattern, whereas this is not true for
edit distance.
We therefore define $\numnodesp{l}{P}$ to be the number of enumerated strings
of length $l$ when performing the search $S$ on a pattern $P$ of length $m$,
partitioned by $\partition$, and a text $\hat{T}$.
We also define $\numnodes{l}$ to be the expectation of $\numnodesp{l}{P}$,
where $P$ is chosen randomly.

We next show how to compute values $\numnodes{l}$.
We begin by giving an algorithm for computing $\numnodesp{l}{P}$ for
some fixed $P$.
Build a non-deterministic automaton $\automaton{P}$ that recognizes the
set of strings that are within edit distance at most $k$ to $P$, and the
locations of the errors satisfy the requirements of the
search~\cite{MihovSchulz,KarkkainenN07}
(see Figure~\ref{fig:NFA} for an example).
For a state $q$ and a string $B$, denote by $\transition{P}(q,B)$ the set of all
states $q'$ for which there is a path in $\automaton{P}$ from $q$ to $q'$
such that the concatenation of the labels on the path is equal to $B$.
For a set of states $Q$ and a string $B$,
$\transition{P}(Q,B)=\cup_{q\in Q} \transition{P}(q,B)$.
Clearly, $\numnodesp{l}{P}$ is equal to the number of strings $B$ of length $l$
for which $\transition{P}(q_0,B)\neq \emptyset$,
where $q_0$ is the initial state.
Let $\numnodesqp{l}{Q}{P}$ be the number of strings $B$ of length $l$ for which
$\transition{P}(q_0,B)=Q$.
The values of $\numnodesqp{l}{Q}{P}$ can be computed using dynamic programming
and the following recurrence.
\[
\numnodesqp{l}{Q}{P} = \sum_{c\in \Sigma}
\sum_{Q': \transition{P}(Q',c)=Q} \numnodesqp{l-1}{Q'}{P}.
\]
The values $\numnodesqp{l}{Q}{P}$ gives the values of
$\numnodesp{l}{P}$, since by definition,
\[
\numnodesp{l}{P} = \sum_{Q} \numnodesqp{l}{Q}{P},
\]
where the summation is done over all non-empty sets of states $Q$.

Note that for a string $B$ of length $l$, set $\transition{P}(q_0,B)$
is a subset of a set of $(k+1)^2$ states that depends on $l$. This set,
denoted $\states{l}$, includes the $l+1$-th state in the first row of the
automaton, states $l,l+1,l+2$ on the second row,
states $l-1,l,\ldots,l+3$ on the third row, and so on (see Figure~\ref{fig:NFA}).
The size of $\states{l}$ is $1+3+5+\cdots+(2k+1) = (k+1)^2$.
Therefore, the number of sets $Q$ for which $\numnodesqp{l}{Q}{P}>0$ is at most
$2^{(k+1)^2}$.
If $(k+1)^2$ is small enough, a state can be encoded in one machine word,
and the computation of $\transition{P}(Q',c)$ can be done in constant time using
precomputed tables.
Thus, the time for computing all values of $\numnodesqp{l}{Q}{P}$ is
$O(2^{k^2}\sigma m)$.

\begin{figure}
\centering
\includegraphics[scale=0.75]{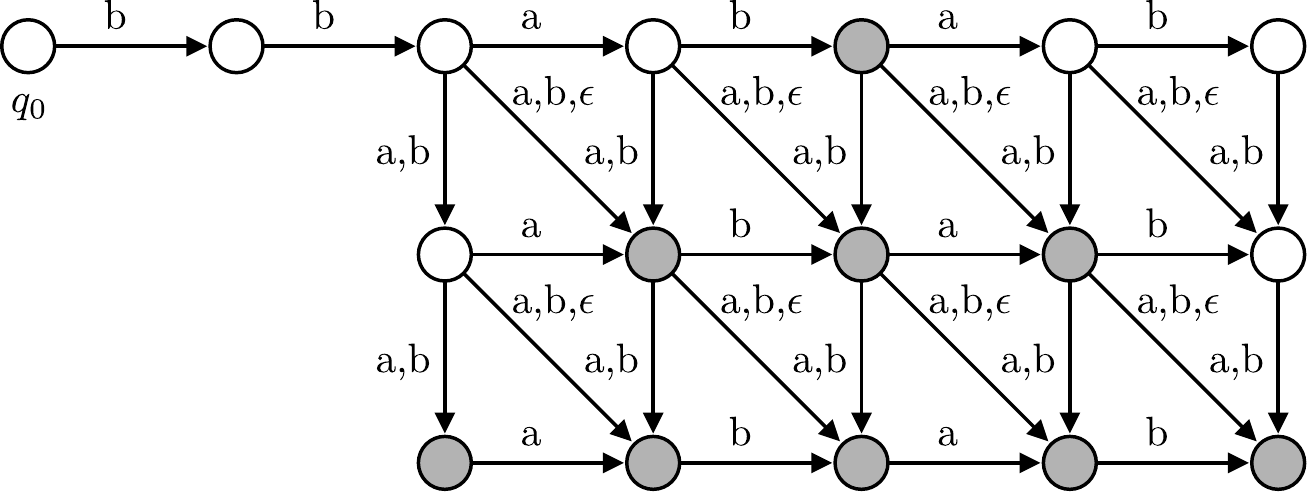}
\caption{Non-deterministic automaton corresponding
to the search $S=(12,00,02)$ and pattern $P=\mathrm{bbabab}$ over the alphabet
$\Sigma=\{\mathrm{a},\mathrm{b}\}$.
A path from the initial state $q_0$ to the state in the $i$-th row and
$j$-column of the automaton correspond to a string with edit distance
$i-1$ to $\substr{P}{1}{j-1}$.
The nodes of the set $\states{4}$ are marked by gray.\label{fig:NFA}}
\end{figure}

Now consider the problem of computing the values of $\numnodes{l}$.
Observe that for $Q\subseteq \states{l}$, the value of
$\transition{P}(Q,c)$ depends on the characters of $\substr{P}{l-k+1}{l+k+1}$,
and does not depend on the rest of the characters of $P$.
Our algorithm is based on this observation.
For an integer $l$, a set $Q \subseteq \states{l}$, and
a string $P'$ of length $2k+1$, define
\[ \numnodesqpb{l}{Q}{P'} = 
 \sum_{P:\substr{P}{l-k+1}{l+k+1}=P'} \numnodesqp{l}{Q}{P}.
\]
Then,
\[
 \numnodesqpb{l}{Q}{P'} = \sum_{c'\in \Sigma} \sum_{c\in \Sigma}
 \sum_{Q': \transition{P_c}(Q',c)=Q} \numnodesqpb{l-1}{Q'}{P'_c},
\]
where $P'_c = c'\substr{P'}{1}{2k}$, and
$P_c$ is a string satisfying $\substr{P_c}{(l-1)-k+1}{(l-1)+k+1} = P'_c$
(the rest of the characters of $P_c$ can be chosen arbitrarily).

From the above, the time complexity for computing
$\numstringsedit(S,\partition,\sigma,n)$ is
$O(|\scheme| 2^{k^2}\sigma^{2k+3} m)$.
Therefore, our approach is practical only for small values of $k$.

\iffull
\subsection{Uneven partitions}\label{sec:uneven}
In Section~\ref{sec:bidirectional}, we provided an informal explanation why
partitioning the pattern into unequal parts may be beneficial. 
We now provide a formal justification for this. 
To this end, we replace (\ref{eq:a2}) by an even simpler estimator of
$\numstrings(S,\partition,\sigma,n)$:
\begin{equation}
\numstringspp(S,\partition,\sigma,n) = \sum_{l = 1}^{\lceil\log_{\sigma}
n\rceil} \numnodes{l}.
\end{equation}

As an example, consider scheme $\schemelam$.
Denote by $x_1$, $x_2$, $x_3$ the lengths of the parts
in a partition $\partition$ of $P$ into 3 parts.
\ifextra
It is straightforward to give closed form formulas for the values of 
$\numnodes{l}$ for the searches of the scheme.
For example, for $S_f$ we have
\begin{align*}
\numnodesd{l}{0} &= 1 \\
\numnodesd{l}{1} &= \begin{cases}
0 & \text{if }l \leq x_1 \\
(\sigma-1)(l-x_1) & \text{otherwise}
\end{cases}\\
\numnodesd{l}{2} &= \begin{cases}
0 & \text{if }l \leq x_1 \\
(\sigma-1)^2 (l-x_1)(l-x_1-1)/2 & \text{otherwise}
\end{cases}
\end{align*}
(if $l > x_1$ then $\numnodesd{l}{2}
 =(\sigma-1)((\sigma-1)+2(\sigma-1)+\cdots+(l-x_1-1)(\sigma-1))$)
and therefore

\begin{equation}\label{eq:nl}
\numnodes{l} = \begin{cases}
1
 & \text{if }l \leq x_1\\
\frac{(\sigma-1)^2}{2}(l-x_1)^2+
 \left((\sigma-1)-\frac{(\sigma-1)^2}{2}\right)(l-x_1)+1
 & \text{otherwise}
\end{cases}
\end{equation}
For the backward search $S_b = (321,012,000)$,
\begin{align*}
\numnodesd{l}{0} &= 1 \\
\numnodesd{l}{1} &= \begin{cases}
0 & \text{if }l \leq x_3 \\
(\sigma-1)(l-x_3) & \text{otherwise}
\end{cases}\\
\numnodesd{l}{2} &= \begin{cases}
0 & \text{if }l \leq x_3 \\
(\sigma-1)^2 (l-x_1)(l-x_1-1)/2 & \text{otherwise}
\end{cases}
\end{align*}
(if $l > x_2+x_3$ then 
$\numnodesd{l}{2}=(\sigma-1)(x_2(\sigma-1)+(x_2+1)(\sigma-1)+\cdots+(l-x_2+x_3-1)(\sigma-1))$)
and therefore
\begin{equation}\label{eq:nl2}
\numnodes{l} = \begin{cases}
1
 & \text{if }l \leq x_3\\
(\sigma-1)(l-x_3)+1
 & \text{if }x_3 < l \leq x_2+x_3\\
\begin{aligned}
& \frac{(\sigma-1)^2}{2}(l-x_2-x_3)^2\\
& \quad + \left((\sigma-1)-(\sigma-1)^2(\frac{1}{2}+x_2)\right)
  (l-x_2-x_3)\\
& \quad + (\sigma-1)x_2+1
\end{aligned}
 & \text{otherwise}
\end{cases}
\end{equation}
Therefore
\else
It is straightforward to give closed formulas for 
$\numstringspp(S,\partition,\sigma,n)$ for each search of $\schemelam$.
For example,
\fi
\[
\numstringspp(S_f,\partition,\sigma,n) = \begin{cases}
\logn & \text{if }\logn \leq x_1\\
c_1(\logn-x_1)^3 + c_2(\logn-x_1)^2 + c_3(\logn-x_1) + \logn
& \text{otherwise}
\end{cases}
\]
where
$\logn = \lceil\log_{\sigma} n\rceil$,
$c_1 = (\sigma-1)^2/6$,
$c_2 = (\sigma-1)/2 $, and
$c_3 = -(\sigma-1)^2/6+(\sigma-1)/2$.
\ifextra
Similarly,
\[
\numstringspp(S_b,\sigma,n) = \begin{cases}
\logn & \text{if }\logn \leq x_3\\
c'_1(\logn-x_3)^2 + c'_2(\logn-x_3) + \logn
& \text{if }x_3 < \logn \leq x_2+x_3 \\
\begin{aligned}
&  c'_3(\logn-x_2-x_3)^3 + (c'_4+c'_5x_2)(\logn-x_2-x_3)^2\\
&+ (c'_6+c'_7x_2)(\logn-x_2-x_3) + \logn
\end{aligned}
& \text{otherwise}
\end{cases}
\]
where
$c'_1=(\sigma-1)^2/2$,
$c'_2=(\sigma-1)/2$,
$c'_3 = (\sigma-1)^2/6$,
$c'_4 = (\sigma-1)/2$,
$c'_5 = (\sigma-1)^2/2$,
$c'_6 = -(\sigma-1)^2/6+(\sigma-1)/2$, and
$c'_7 = (\sigma-1)^2/2+(\sigma-1)$.
\else
Similar formulas can be given for $S_b$ and $S_{bd}$.
\fi
If $x_1$, $x_2$, and $x_3$ are close to $m/3$ and $\logn < m/3$ then
$\numstringspp(\schemelam,\partition,\sigma,n,m) = 3\logn$ and an
equal sized partition is optimal in this case.
However, if $m/3 < \logn < 2m/3$, then
\begin{align*}
\numstringspp(\schemelam,&\partition,\sigma,n) =
   c_1(\logn-x_1)^3
 + c_2(\logn-x_1)^2
 + c_3(\logn-x_1)\\
& 
 + c'_1(\logn-x_3)^2
 + c'_2(\logn-x_3)
 + c''_1(\logn-x_2)^2
 + c''_2(\logn-x_2)
 +3\logn.
\end{align*}
It is now clear why the equal sized partition is not optimal in this case.
The degree of $\logn-x_1$ in the above polynomial is 3, while
the degrees of $\logn-x_2$ and $\logn-x_3$ are 2.
Thus, if $x_1 = x_2 = x_3 = m/3$, decreasing $x_2$ and $x_3$ by, say 1,
while increasing $x_1$ by 2 reduces the value of the polynomial.
\fi

\subsection{Computing an optimal partition}
In this Section, we show how to find an optimal partition for a given
search scheme $\scheme$ and a given number of parts $p$.
An optimal partition can be naively found by enumerating
all $\binom{m-1}{p-1}$ possible partitions, 
and for each partition $\partition$, computing
$\numstringsp(\scheme,\partition,\sigma,n)$.
We now describe a more efficient dynamic programming algorithm.

We define an optimal partition to be a partition that maximizes
$\numstrings(\scheme,\partition,\sigma,n)$.
Let $\lognc=\lceil\log_{\sigma} n\rceil+c_\sigma$.
If $m \geq p \lognc$, then any partition in which all parts are of size at least
$\lognc$ is an optimal partition.
Therefore, assume for the rest of this section that $m < p \lognc$.
We say that a partition $\partition$ is \emph{bounded} if the sizes of the parts
of $\partition$ are at most $\lognc$.
If $\partition$ is not bounded, we can transform it into a bounded
partition by decreasing the sizes of parts which are larger than $\lognc$
and increasing the sizes of parts which are smaller that $\lognc$.
This transformation 
can only decrease the value of
$\numstrings(\scheme,\partition,\sigma,n)$.
Therefore, there exists an optimal partition which is bounded.
Throughout this section we will consider only bounded partitions.
For brevity, we will use the term partition instead of bounded partition.

Our algorithm takes advantage of the fact that the value of
$\numstringsp(\scheme,\partition,\sigma,n)$ does not depend on the entire
partition $\partition$, but only on the partition of a substring of $P$
of length $\lognc$ induced by $\partition$.
More precisely, consider a fixed $S=(\pi,L,U)\in \scheme$.
By definition,
$\numstringsp(S,\partition,\sigma,n)$ depends on
the values $\numnodes{1},\ldots,\numnodes{\lognc}$ (the number of nodes in
levels $1,\ldots,\logn$ in the trie that correspond to the search $S$).
From Section~\ref{sec:estimation}, these values depend on
the strings $L$ and $U$ which are fixed, and on the string
$\substr{\pi_{\partition}}{1}{\lognc}$.
The latter string depends on $\substr{\pi}{1}{i_{\partition,\pi}}$,
where $i_{\partition,\pi}$ is the minimum index such that
$\sum_{j=1}^{i_{\partition,\pi}} X[\pi(j)] \geq \lognc$
and on the values 
$\partition[\pi(1)],\ldots,\partition[\pi(i_{\partition,\pi})]$.

The algorithm works by going over the prefixes of $P$ in increasing length
order. For each prefix $P'$, it computes a set of partitions of $P'$
such that at least one partition in this set can be extended to an optimal
partition of $P$.
In order to reduce the time complexity, the algorithm needs to identify
partitions of $P'$ that cannot be extended into an optimal partition of $P$.
Consider the following example.
Suppose that $m = 13$, $p = 5$, $\lognc = 4$ and $\scheme = \{S_1,S_2,S_3\}$,
where the $\pi$-strings of $S_1,S_2,S_3$ are $\pi^1 = 12345$, $\pi^2 = 32451$,
and $\pi^3 = 43215$, respectively.
Consider a prefix $P'=\substr{P}{1}{8}$ of $P$, and let
$\ppartition_1,\ppartition_2$ be two partitions of $P'$,
where the parts in $\ppartition_1$ are of sizes 3,3,2,
and the parts in $\ppartition_1$ are of sizes 4,2,2.
Note that $\ppartition_1$ and $\ppartition_2$ have the same number of parts,
and they induce the same partition on $\substr{P}{8-\lognc+1}{8}=\substr{P}{5}{8}$.
We claim that one of these two partitions is always at least as good
as the other for every extension of both partitions to a partition of $P$.
To see this, let $\spartition$ denote a partition of $\substr{P}{9}{13}$ into
two parts, and consider the three searches of $\scheme$.

\begin{enumerate}
\item
For search $S_1$ we have that 
$\substr{\pi^1_{\ppartition_1 \cup \spartition}}{1}{\lognc} = 1112$
for every $\spartition$, and 
$\substr{\pi^1_{\ppartition_2 \cup \spartition}}{1}{\lognc} = 1111$
for every $\spartition$.
It follows that the value of
$\numstringsp(S_1,\ppartition_1 \cup \spartition,\sigma,n)$ is the same
for every $\spartition$, and the value of 
$\numstringsp(S_1,\ppartition_2 \cup \spartition,\sigma,n)$ is the same
for every $\spartition$.
These two values can be equal or different.

\item
For the search $S_2$ we have that 
$\substr{\pi^2_{\ppartition_1 \cup \spartition}}{1}{\lognc} = 
 \substr{\pi^2_{\ppartition_2 \cup \spartition}}{1}{\lognc} = 3322$.
It follows that
$\numstringsp(S_2,\ppartition_1 \cup \spartition,\sigma,n) =
 \numstringsp(S_2,\ppartition_2 \cup \spartition,\sigma,n)$
for all $\spartition$ and this common value does not depend on $\spartition$.

\item
For the search $S_3$ we have that 
$\substr{\pi^3_{\ppartition_1 \cup \spartition}}{1}{\lognc} =
 \substr{\pi^3_{\ppartition_2 \cup \spartition}}{1}{\lognc}$
for every $\spartition$.
For example, if $\spartition$ is a partition of $\substr{P}{9}{13}$
into parts of sizes 2,2 then 
$\substr{\pi^3_{\ppartition_1 \cup \spartition}}{1}{\lognc} =
 \substr{\pi^3_{\ppartition_2 \cup \spartition}}{1}{\lognc} = 4433$.
It follows that
$\numstringsp(S_3,\ppartition_1 \cup \spartition,\sigma,n) =
 \numstringsp(S_3,\ppartition_2 \cup \spartition,\sigma,n)$
for every $\spartition$.
This common value depends on $\spartition$.

\end{enumerate}
We conclude that either 
$\numstringsp(\scheme,\ppartition_1 \cup \spartition,\sigma,n) < 
 \numstringsp(\scheme,\ppartition_2 \cup \spartition,\sigma,n)$
for every $\spartition$,
or $\numstringsp(\scheme,\ppartition_1 \cup \spartition,\sigma,n) \geq 
 \numstringsp(\scheme,\ppartition_2 \cup
 \spartition,\sigma,n)$
for every $\spartition$.

We now give a formal description of the algorithm.
We start with some definitions.
For a partition $\ppartition$ of a substring $P'=\substr{P}{m''}{m'}$ of 
pattern $P$, we define the following quantities:
$\plen{\ppartition}$ is the length of $P'$,
$\plast{\ppartition}$ is the length of the last part of $\ppartition$,
$\pnumparts{\ppartition}$ is the number of parts in $\ppartition$,
and $\prank{\ppartition}$ is the left-to-right rank of the part of
$\ppartition$ containing $P'[m'-\lognc+1]$.
Let $\pprefix{\ppartition}$ be the partition of 
$\substr{P}{m''}{m'-\plast{\ppartition}}$ of $P'$ that is
composed from the first $\pnumparts{\ppartition}-1$ parts of $\ppartition$.
For the example above,
$\plen{\ppartition_1} = 8$,
$\plast{\ppartition_1} = 2$,
$\pnumparts{\ppartition_1} = 3$,
$\prank{\ppartition_1} = 2$, and
$\pprefix{\ppartition_1}$ is a partition of $\substr{P}{1}{6}$ with
parts sizes $3,3$.

For a partition $\ppartition$ of a prefix $P'$ of $P$,
$\scheme(\ppartition)$ is a set containing every search $S\in \scheme$
such that $\prank{\ppartition}$ appears before
$\pnumparts{\ppartition}+1$ in the $\pi$-string of $S$.
If the length of $P'$ is less than $\lognc$ we define 
$\scheme(\ppartition)=\emptyset$, and if $P' = P$ we define
$\scheme(\ppartition)=\scheme$.
For the example above, $\scheme(\ppartition_1) = \{S_1,S_2\}$.

Let $\ppartition_1$ be a partition of a substring
$P_1=\substr{P}{i_1}{j_1}$ of $P$,
and $\ppartition_2$ be a partition of a substring $P_2=\substr{P}{i_2}{j_2}$.
We say that $\ppartition_1$ and $\ppartition_2$ are \emph{compatible}
if these partitions induce the same partition on the common substring
$P'=\substr{P}{\max(i_1,i_2)}{\min(j_1,j_2)}$.
For example, the partition of $\substr{P}{4}{6}$ into parts of sizes $1,2$
is compatible with the partition of $\substr{P}{1}{6}$ into parts of sizes
$2,2,2$.

\begin{lemma}\label{lem:optimal-1}
Let $\ppartition$ be a partition of a prefix of $P$ of length at least $\lognc$.
Let $S\in \scheme(\ppartition)$ be a search.
The value $\numstringsp(S,\partition,\sigma,n)$ is the same for every
partition $\partition$ of $P$ whose first $\pnumparts{\ppartition}$ parts
match $\ppartition$.
\end{lemma}
\begin{proof}
Let $i'$ be the index such that $\pi(i')=\pnumparts{\ppartition}+1$.
Since $\prank{\ppartition}$ appears before
$\pnumparts{\ppartition}+1$ in string $\pi$,
from the connectivity property of $\pi$ we have that
(1) Every value in $\pi$ that appears before $\pnumparts{\ppartition}+1$
is at most $\pnumparts{\ppartition}$.
In other words, $\pi(i) \leq \pnumparts{\ppartition}$ for every $i < i'$.
(2) $\prank{\ppartition},\ldots,\pnumparts{\ppartition}$ appear before
$\pnumparts{\ppartition}+1$ in $\pi$.
By the definition of $\prank{\ppartition}$,
$\sum_{j=\prank{\ppartition}}^{\pnumparts{\ppartition}} X[j] \geq \lognc$.
Therefore, $i_{\partition,\pi} < i'$ and 
$\pi(1),\ldots,\pi(i_{\partition,\pi}) \leq \pnumparts{\ppartition}$.
Thus, string $\substr{\pi}{1}{i_{\partition,\pi}}$ and values
$\partition[\pi(1)],\ldots,\partition[\pi(i_{\partition,\pi})]$
are the same for every partition $\partition$ that satisfies the requirement of
the lemma.
\end{proof}

For a partition $\ppartition$ of a prefix of $P$ of length at least $\lognc$,
define $v(\ppartition)$ to be 
$\sum_{S\in \scheme(\ppartition)} \numstringsp(S,\partition,\sigma,n)$,
where $\partition$ is an arbitrary partition of $P$ whose first
$\pnumparts{\ppartition}$ parts match $\ppartition$
(the choice of $\partition$ does not matter due to Lemma~\ref{lem:optimal-1}).
For a partition $\ppartition$ of a prefix of $P$ of length less than $\lognc$,
$v(\ppartition) = 0$.
Define
\[\Delta(\ppartition) = v(\ppartition) - v(\pprefix{\ppartition})
= \sum_{S\in \scheme(\ppartition)\setminus\scheme(\pprefix{\ppartition})}
\numstringsp(S,\partition,\sigma,n).
\]

\begin{lemma}\label{lem:optimal-2}
Let $\spartition$ be a partition of a substring $\substr{P}{m''}{m'}$
such that $\pnumparts{\spartition} \geq 2$ and
$\plen{\pprefix{\spartition}} = \min(\lognc,m'-\plast{\ppartition})$.
Let $p'\geq \pnumparts{\spartition}$ be an integer.
The value of $\Delta(\ppartition)$ is the same for
every partition $\ppartition$ of $\substr{P}{1}{m'}$ with $p'$ parts
that is compatible with $\spartition$.
\end{lemma}
\begin{proof}
We assume $\lognc < m'-\plast{\ppartition}$ (the case
$\lognc \geq m'-\plast{\ppartition}$ is similar).
Since $\plen{\pprefix{\spartition}} = \min(\lognc,m'-\plast{\ppartition})$,
the set $\scheme(\ppartition)\setminus\scheme(\pprefix{\ppartition})$ is the
same for every partition $\ppartition$ of $\substr{P}{1}{m'}$ with $p'$ parts
that is compatible with $\spartition$.
For a search $S=(\pi,L,U)$ in this set, $\prank{\ppartition}$ appears before
$\pnumparts{\ppartition}+1$ in $\pi$,
and $\pnumparts{\ppartition}$ appears before $\prank{\pprefix{\ppartition}}$.
Let $i = i_{\partition,\pi}$, 
where $\partition$ is an arbitrary partition of $P$ whose first
$\pnumparts{\ppartition}$ parts are the parts of $\ppartition$.
We obtain that 
$\prank{\pprefix{\ppartition}} \leq
\pi(1),\ldots,\pi(i) \leq \pnumparts{\ppartition}$,
and the lemma follows.
\end{proof}

For $\spartition,p'$ that satisfy the requirements of
Lemma~\ref{lem:optimal-2}, let $\Delta(\spartition,p')$ denote the value of
$\Delta(\ppartition)$, where $\ppartition$ is an arbitrary partition of
$\substr{P}{1}{m'}$ with $p'$ parts that is compatible with $\spartition$.

For $m'\leq m$, $p' \leq p$, and a partition $\spartition$ of
$\substr{P}{\max(m'-\lognc+1,1)}{m'}$ with at most $p'$ parts,
let $\optv{m'}{p'}{\spartition}$ be the minimum value of $v(\ppartition)$,
where $\ppartition$ is a partition of $\substr{P}{1}{m'}$ into $p'$ parts
that is compatible with $\spartition$.
\begin{lemma}\label{lem:optimal-3}
For $m' \leq m$, $2 \leq p' \leq p$, and a partition $\spartition$ of
$\substr{P}{\max(m'-\lognc+1,1)}{m'}$ with at most $p'$ parts,
\[
\optv{m'}{p'}{\spartition} = \min_{\spartition'} \left(
  \optv{m'-\plast{\spartition'}}{p'-1}{\pprefix{\spartition'}} +
  \Delta(\spartition',p') 
\right)
\]
where the minimum is taken over all partitions $\spartition'$
of a substring $\substr{P}{m''}{m'}$ of $P$ that satisfy the following:
(1) $\spartition'$ is compatible with $\spartition$,
(2) $2 \leq \pnumparts{\spartition'} \leq p'$,
(3) $\plen{\pprefix{\spartition'}} = \min(\lognc, m'-\plast{\spartition'})$,
(4) $\pnumparts{\spartition} = p'$ if $m'' = 1$.
\end{lemma}
An algorithm for computing the optimal partition follows from
Lemma~\ref{lem:optimal-3}.
The time complexity of the algorithm is
$O\big( {(|\scheme| k\lognc+m)}
\sum_{j=1}^{\min(p-1,\lognc)}(p-j)\binom{\lognc-1}{j-1}\big)$, where
$|\scheme| k\lognc
\sum_{j=1}^{\min(p-1,\lognc)}(p-j)\binom{\lognc-1}{j-1}$ is time for
computing $\Delta$ values, and $O\big(m
\sum_{j=1}^{\min(p-1,\lognc)}(p-j)\binom{\lognc-1}{j-1}\big)$ is time
for computing $v$ values. 

\section{Properties of optimal search schemes}\label{sec:design}

Designing an efficient search scheme for a given set of parameters consists
of
\begin{inparaenum}[(1)]
\item choosing a number of parts,
\item choosing searches, 
\item choosing a partition of the pattern.
\end{inparaenum}
While it is possible to enumerate all possible choices, and 
evaluate the efficiency of the resulting scheme using
Section~\ref{sec:estimation}, this is generally infeasible due to a large
number of possibilities.
It is therefore desirable to have a combinatorial characterization of
optimal search schemes.

The \emph{critical string} of a search scheme $\scheme$ is the
lexicographically maximal $U$-string of a search in $\scheme$.
A search of $\scheme$ is \emph{critical} if its $U$-string is equal to the
critical string of $\scheme$.
For example, the critical string of $\schemelam$ is $022$, and 
$S_f$ is the critical search. 
For typical parameters, critical searches of a search scheme constitute
the bottleneck.
Consider a search scheme $\scheme$, and assume that the $L$-strings
of all searches 
contain only zeros.
Assume further that the pattern is partitioned into 
equal-size parts.
Let $\ell$ be the maximum index such that for every search $S\in\scheme$
and every $i \leq \ell$, $U[i]$ of $S$ is no larger than
the number in position $i$ in the critical string of $\scheme$.
From Section~\ref{sec:analysis}, the number of strings enumerated by a search
$S\in \scheme$ depends 
mostly
on the prefix of the $U$-string of $S$
of length $\lceil \lceil\log_{\sigma} n\rceil/(m/p)\rceil$.
Thus, if $\lceil \lceil\log_{\sigma} n\rceil/(m/p)\rceil \leq \ell$,
a critical search enumerates an equal or greater number of strings than
a non-critical search.

We now consider the problem of designing a search scheme whose critical string
is minimal.
Let $\alpha(k,p)$ denote the lexicographically minimal critical string of a
$k$-mismatch search scheme that partitions the pattern into $p$ parts.
The next theorems give the values of $\alpha(k,k+2)$ and $\alpha(k,k+1)$.
\iffull
We need the following definition.
A string over the alphabet of integers is called \emph{simple}
if it contains a substring of the form $01^j0$ for $j\geq 0$.
\begin{lemma}\label{lem:simple}
\begin{itemize}
\item[(i)] Every string $A$ of weight $k$ and length at least $k+2$ is
  simple.
\item[(ii)] If $A$ is a non-simple string of weight $k$ and length $k+1$ then
$A[1]\leq 1$, $A[k+1]\leq 1$, and $A[i]\leq 2$ for all $2\leq i\leq k$.
Moreover, there are no two consecutive $2$'s in $A$.
\end{itemize}
\end{lemma}
\begin{proof}
\emph{(i)}
The proof is by induction on $k$.
It is easy to verify that the lemma holds for $k=0$.
Suppose we proved the lemma for $k'<k$.
Let $A$ be a string of weight $k$ and length $p \geq k+2$.
If $A[1]\geq 1$ then by the induction hypothesis $\substr{A}{2}{p}$ is simple,
and therefore $A$ is simple.
Suppose that $A[1] = 0$.
Let $i>1$ be the minimum index such that $A[i] \neq 1$
($i$ must exist due to the assumption that $p\geq k+2$).
If $A[i] = 0$ then we are done.
Otherwise, we can use the induction hypothesis on $\substr{A}{i+1}{p}$
and obtain that $A$ is simple.

\emph{(ii)}
Let $A$ be a non-simple string of weight $k$ and length $k+1$.
If $A[1]\geq 2$ then $A'=\substr{A}{2}{k+1}$ has weight $k-A[1]\leq k-2$ and
length $k$, and thus by
\emph{(i)} we obtain that $A'$
is simple, contradicting the assumption that $A$ is non-simple.
Similarly, $A[k+1]$ cannot be greater than $1$.
For $2\leq i\leq k$, if $A[i] \geq 3$ then either $\substr{A}{1}{i-1}$ or
$\substr{A}{i+1}{k+1}$ satisfies the condition of 
\emph{(i)}.
Similarly, if $A[i]=A[i+1]=2$ then either $\substr{A}{1}{i-1}$ or
$\substr{A}{i+2}{k+1}$ satisfies the condition of 
\emph{(i)}.
\end{proof}
\else
We omit the proofs due to space limitations.
\fi

We use the following notation.
For two integers $i$ and $j$, $\seq{i}{j}$ denotes the string
$i(i+1)(i+2)\cdots j$ if $i \leq j$, and the empty string if $i > j$.
Moreover, $\seqr{i}{j}$ denotes the string $i(i-1)(i-2)\cdots j$ if $i \geq j$,
and the empty string if $i < j$.

\begin{theorem}\label{thm:alpha}
$\alpha(k,k+1) = 013355\cdots kk$ for every odd $k$, and
$\alpha(k,k+1) = 02244\cdots kk$ for every even $k$.
\end{theorem}
\iffull
\begin{proof}
We first give an upper bound on $\alpha(k,k+1)$ for odd $k$.
We build a search scheme as follows.
The scheme contains searches 
$S_{k,i,j} = (\seq{i}{k+2}\seqr{i-1}{1},0\cdots 0,\seq{0}{j}jk\cdots k)$
for all $i$ and $j$, which cover all simple strings
of weight $k$ and length $k+1$.
In order to cover the non-simple strings, the scheme contains the following
searches.
\begin{enumerate}
\item $S_{k,i,j}^{1}=(\seq{i}{k+1}\seqr{i-1}{1},
0\cdots 0,013355\cdots jj(j+1)k\cdots k)$ for every odd $3 \leq j\leq k$
(for $j=k$, the $U$-string is $013355\cdots kk$).
\item $S_{k,i,j}^{2}=(\seqr{i}{1}\seq{i+1}{k+1},
0\cdots 0,013355\cdots jj(j+1)k\cdots k)$ for every odd $3 \leq j\leq k$
(for $j=k$, the $U$-string is $013355\cdots kk$).
\end{enumerate}

Let $A$ be a non-simple string of weight $k$ and length $k+1$.
By Lemma~\ref{lem:simple}, $A=X0A_10A_20\cdots 0 A_d0Y$ where each of $X$ and $Y$
is either string $1$ or empty string, and each $A_i$ is either
$2$, $12$, $21$, or $121$.
A string $A_i$ is called a \emph{block} of type~1, 2, or~3 if $A_i$ is
equal to $12$, $21$, or $121$, respectively.
Let $B_1,\ldots,B_{d'}$ be the blocks of type~1 and type~2, from left to right.

We consider several cases.
The first case is when $X$ and $Y$ are empty strings, and $B_1$ is of type~1.
Since the weight of $A$ is odd, it follows that $d'$ is odd.
If $A$ has no other blocks, $A$ is covered by search
$S_{k,i,k}^{1}$, where $i+1$ is the index in $A$ in which $B_1$ starts.
Otherwise, if $B_2$ is of type~1, then $A$ is covered by search
$S_{k,i,j}^{1}$, where $i+1$ is the index in $A$ in which $B_1$ starts,
and $i+j+1$ is the index in which the first block to the right of $B_1$ starts
(this block is either $B_2$, or a block of type~3).
Now suppose that $B_2$ is of type~2.
If $B_3$ is of type~2, then $A$ is covered by search
$S_{k,i,j}^{2}$, where $i-1$ is the index in $A$ in which $B_3$ ends,
and $i-j-1$ is the index in which the first block to the left of $B_3$ ends.
By repeating these arguments, we obtain that $A$ is covered unless
the types of $B_1,\ldots,B_{d'}$ alternate between type~1 and type~2.
However, since $d'$ is odd, $B_{d'}$ is of type~1, and in this case
$A$ is covered by $S_{k,i,j}^{1}$,
where $i+1$ is the index in $A$ in which $B_1$ starts,
and $k-j$ is the index in which the first block to the left of $B_1$ ends.

Now, if $X$ is empty string and $Y=1$, define a string $A'=A20$.
By the above, $A'$ is covered by some search $S_{k+2,i,j}^{j'}$.
Then, $A$ is covered by either $S_{k,i,j}^{j'}$ or $S_{k,i,j-2}^{j'}$.
The same argument holds for the case when $X=1$.
The proof for the case when $B_1$ is of type~2 is analogous and thus omitted.

The lower bound on $\alpha(k,k+1)$ for odd $k$ is obtained by considering
the string $A=012020\cdots 20$. The $U$-string of a search that covers $A$
must be at least $013355\cdots kk$.

We next give an upper bound on $\alpha(k,k+1)$ for even $k$.
We define $k$-mismatch search schemes $\scheme_k$ recursively.
For $k=0$, $\scheme_0$ consists of a single search $S_{0,1}=(1,0,0)$.
For $k\geq 2$, $\scheme_k$ consists of the following searches.
\begin{enumerate}
\item For every search $S_{k-2,i}=(\pi,0\cdots 0,U)$ in $\scheme_{k-2}$,
$\scheme_k$ contains a search
$S_{k,i}=(\pi\cdot k(k+1), 0\cdots 0, U\cdot kk)$.
\item A search
$S_{k,k} = (\seqr{k+1}{1}, 0\cdots 0, 01kk\cdots k)$.
\item A search
$S_{k,k+1} = (k(k+1)\seqr{k-1}{1}, 0\cdots 0, 01kk\cdots k)$.
\end{enumerate}
Note that the critical string of $\scheme_k$ is $02244\cdots kk$
corresponding to item 1 above.
We now claim that all number strings of length $k+1$ and weight at most $k$
are covered by the searches of $\scheme_k$.
The proof is by induction on $k$.
The base $k=0$ is trivial.
Suppose the claim holds for $k-2$.
Let $A$ be a number string of length $k+1$ and weight $k'\leq k$.
If $A[k]+A[k+1]\leq 1$, then $A$ is covered by either $S_{k,k}$ or $S_{k,k+1}$.
Otherwise, the weight of $A' = \substr{A}{1}{k-1}$ is at most
$k'-2$.
By induction,
$A'$ is covered by some search
$S_{k-2,i}$. Then search $S_{k,i}$ covers $A$.

To prove that $\alpha(k,k+1) \geq 02244\cdots kk$ for even $k$,
consider the string $A=0202\cdots 020$.
It is easy to verify that the $U$-string of a search that covers $A$
must be at least $02244\cdots kk$.
\end{proof}
\fi
\begin{theorem}\label{thm:alpha2}
$\alpha(k,k+2) = 0123\cdots (k-1)kk$ for every $k\geq 1$.
\label{theorem2}
\end{theorem}
\iffull
\begin{proof}
We first give an upper bound on $\alpha(k,k+1)$.
We build a $k$-mismatch search scheme $\scheme$ that contains
searches
$S_{k,i,j} = (\seq{i}{k+2}\seqr{i-1}{1},0\cdots 0,\seq{0}{j}jk\cdots k)$
for all $i$ and $j$.
Let $A$ be a string of weight $k$ and length $k+2$.
By Lemma~\ref{lem:simple} there are indices $i$ and $j$ such that
$\substr{A}{i}{i+j+1}=01^j0$, and therefore $A$ is covered by
$S_{k,i,j}$.

The lower bound is obtained from the string $A=011\cdots 110$.
It is easy to verify that the $U$-string of a search that covers $A$
must be at least $0123\cdots (k-1)kk$.
\end{proof}
\fi

An important consequence of Theorems~\ref{thm:alpha} and
\ref{theorem2} is that for some typical cases, partitioning the pattern
into $k+2$ parts brings an advantage over $k+1$ parts. 
For $k=2$, for example, we have $\alpha(2,3)=022$ while $\alpha(2,4)=0122$.
Since the second element of $0122$ is smaller than that
of $022$, a 4-part search scheme potentially enumerates
less strings than a 3-part scheme.
On the other hand, the average length of a part is smaller when using 4 parts,
and therefore the branching occurs earlier in the searches of a 4-part scheme.
The next section shows that for some parameters, $(k+2)$-part schemes
outperform $(k+1)$-part schemes,
while for other parameters the inverse occurs.

\section{Case studies}\label{sec:experiments}
In this Section, we provide results of several computational
experiments we have performed to analyse practical applicability
of our techniques. 

We designed search schemes for 2, 3 and 4 errors (given in Appendix) using a
greedy algorithm. The algorithm iteratively adds searches to a search
scheme. 
At each step, the algorithm considers the uncovered string $A$ of weight $k$
such that the lexicographically minimal $U$-string that covers $A$
is maximal.
Among the searches that cover $A$ with minimal $U$-string, a search that
covers the maximum number of uncovered strings of weight $k$ is chosen.
The $L$-string of the search is chosen to be lexicographically maximal among
all possible $L$-string that do not decrease the number of uncovered strings.
For each search scheme and each choice of parameters, we computed an
optimal partition. 

\subsection{Numerical comparison of search schemes}
\label{subsec:experimenttheor}
We first performed a comparative estimation of the efficiency of
search schemes using the method of Section~\ref{sec:hamming-estimation}
(case of Hamming distance). More precisely, for a given search scheme $\scheme$,
we estimated the number of strings
$\numstrings(\scheme,\partition,\sigma,n)$ enumerated during the search. 

Results for $2$ mismatches are given in Table~\ref{tab:sigma4} and
Table~\ref{tab:sigma30} for $4$-letter and $30$-letter alphabets
respectively.
Table~\ref{tab:nonuniform} contains estimations for nonuniform letter
distribution.
Table~\ref{tab:k3} contains estimations for $3$
mismatches for $4$-letter alphabet.
%

We first observe that our method provides an advantage
only on a limited range of pattern lengths. This conforms to our analysis 
\iffull
(see Section~\ref{sec:uneven})
\else
(details omitted due to lack of space)
\fi
that implies that our schemes can bring an improvement when $m/(k+1)$ is smaller than
$\log_{\sigma} n$ approximately. 
When $m/(k+1)$ is small, Tables~\ref{tab:sigma4}--\ref{tab:k3} suggest that
using more parts of unequal size can
bring a significant improvement. 
%
For big alphabets (Table~\ref{tab:sigma30}),
we observe a larger gain in efficiency,
due to the fact that values 
$\numnodes{l}$ (see equation~(\ref{eqn:nl})) grow faster when the
alphabet is large, and thus a change in the size of parts can
have a bigger influence on these values.
Moreover, if the probability distribution of letters in both the text and
the pattern is nonuniform,
then we obtain an even larger gain (Table~\ref{tab:nonuniform}), since in this case, the
strings enumerated during the search have a larger probability to
appear in the text than for the uniform distribution. 

For $3$ mismatches and $4$ letters (Table~\ref{tab:k3}), 
we observe a smaller gain, and even a loss for pattern lengths $36$ and
$48$ when shifting from $4$ to $5$ parts.
This is explained by Theorem~\ref{thm:alpha} showing the
difference of critical strings between odd and even numbers of
errors. 
Thus, for $3$ mismatches and $4$ parts, the
critical string is $0133$ 
while for $5$ parts it is $01233$. 
When patterns are not too small, the latter does not lead to an improvement strong enough
to compensate for the decrease of part length. 
Note that the situation is different for even number of errors, where incrementing the number of parts from
$k+1$ to $k+2$ leads to transforming the critical strings from
$0224\cdots$ to $0123\cdots$. 

Another interesting observation is that with $4$ parts, obtained
optimal partitions have equal-size parts, as the $U$-strings
of all searches of the 4-part scheme are all the same (see Appendix). 


These estimations suggest that our techniques can 
bring a
significant gain in efficiency for some parameter ranges, however the
design of a search scheme should be done carefully for each specific
set of parameters.

\begin{table}[!tb]
\caption{Values of $\numstrings(\scheme,\partition,4,4^{16})$ for $2$-mismatch
search schemes, for different pattern lengths $m$. 
Second column corresponds to 
search scheme $\schemelam$ with three equal-size parts, 
the other columns show results for unequal partitions and/or more
parts. The partition used is shown in the second sub-column. 
\label{tab:sigma4}}
\centering
\begin{tabular}{|c|c|c|c|c|c|c|c|}
\hline
$m$ & 3 equal & \multicolumn{2}{|c|}{3 unequal} &
\multicolumn{2}{|c|}{4 unequal} & \multicolumn{2}{|c|}{5 unequal} \\
\hline
24 & 1197 &  1077 & 9,7,8    &  959 & 7,4,4,9   &  939 & 7,1,6,1,9 \\
36 &  241 &   165 & 15,10,11 &  140 & 12,5,7,12 &  165 & 11,1,9,1,14 \\
48 &   53 &    53 & 16,16,16 &   51 & 16,7,9,16 &   53 & 16,1,15,1,15 \\
\hline
\end{tabular}
\end{table}
\begin{table}[!tb]
\caption{Values of $\numstrings(\scheme,\partition,30,30^{7})$ for $2$-mismatch
search schemes.\label{tab:sigma30}}
\centering
\begin{tabular}{|c|c|c|c|c|c|c|c|}
\hline
$m$ & 3 equal & \multicolumn{2}{|c|}{3 unequal} &
\multicolumn{2}{|c|}{4 unequal} & \multicolumn{2}{|c|}{5 unequal} \\
\hline
15 &  846 &  286 & 6,4,5 &  231 & 5,2,3,5 &  286 & 5,1,3,1,5 \\
18 &  112 &  111 & 7,6,5 &   81 & 6,2,4,6 &  111 & 6,1,4,1,6 \\
21 &   24 &   24 & 7,7,7 &   23 & 7,3,4,7 &   24 & 7,1,6,1,6 \\
\hline
\end{tabular}
\end{table}
%
%
\begin{table}[!tb]
\caption{Values of $\numstrings(\scheme,\partition,4,4^{16})$ for $2$-mismatch
search schemes, using a non-uniform letter distribution (one letter with
probability $0.01$ and the rest with probability $0.33$
each).\label{tab:nonuniform}}
\centering
\begin{tabular}{|c|c|c|c|c|c|c|c|}
\hline
$m$ & 3 equal & \multicolumn{2}{|c|}{3 unequal} &
\multicolumn{2}{|c|}{4 unequal} & \multicolumn{2}{|c|}{5 unequal} \\
\hline
24 & 3997 &  3541 & 10,8,6   & 3592 & 6,7,1,10  & 3541 & 6,1,7,1,9 \\
36 &  946 &   481 & 16,10,10 &  450 & 11,6,6,13 &  481 & 10,1,9,1,15 \\
48 &  203 &   157 & 18,15,15 &  137 & 16,7,9,16 &  157 & 15,1,14,1,17\\
\hline
\end{tabular}
\end{table}
%
%
\begin{table}[!tb]
\caption{Values of $\numstrings(\scheme,\partition,4,4^{16})$ for $3$-mismatch
search schemes. Best partitions obtained for 4 parts are equal.\label{tab:k3}}
\centering
\begin{tabular}{|c|c|c|c|c|}
\hline
$m$ & \multicolumn{2}{|c|}{4 equal/unequal} & \multicolumn{2}{|c|}{5 unequal} \\
\hline
24 & 11222 & 6,6,6,6 &  8039 & 4,6,5,1,8 \\
36 &   416 & 9,9,9,9 &   549 & 6,11,5,1,13 \\
48 &   185 & 12,12,12,12 &   213 & 11,11,11,1,14 \\
\hline
\end{tabular}
\end{table}

\subsection{Experiments on genomic data}\label{subsec:experimentpract}

To perform large-scale experiments on genomic sequences, we
implemented our method using the {\sc 2BWT} library
provided by~\cite{LamLTWWY09}
(\url{http://i.cs.hku.hk/2bwt-tools/}). We then experimentally compared
different search schemes, both in terms of running time and average
number of enumerated substrings. Below we only report 
running time, as in all cases, the number of enumerated substrings
produced very similar results. 

\begin{sloppypar}
The experiments were done on the sequence of human chromosome
14 (\emph{hr14}). The sequence is $88\cdot 10^6$ long, with nucleotide distribution
29\%, 21\%, 21\%, 29\%.
Searched patterns were generated as i.i.d.\ sequences. For every search
scheme and pattern length, we ran $10^5$ pattern searches for Hamming distance
and $10^4$ searches for the edit distance.
\end{sloppypar}

\subsubsection{Hamming distance}

For the case of 2 mismatches, we implemented the 3-part and 4-part schemes
(see Appendix), as well as their equal-size-part
versions for comparison. For each pattern length, we computed an
optimal partition, taking into account a non-uniform distribution of
nucleotides. Results are presented in Table~\ref{tab:times2}.

Using unequal parts for 3-part schemes yields a notable time decrease 
for patterns of length $24$ and $33$ (respectively, by 24\% and 16\%). 
Furthermore, we observe that
using unequal part lengths for 4-part schemes is beneficial as
well. For pattern lengths $24$ and $33$, we obtain a speed-up by
27\% and 28\% respectively.
Overall, the experimental results are consistent with numerical
estimations of Section~\ref{subsec:experimenttheor}. 

For the case of 3 mismatches, we implemented 4-part and 5-part schemes
from Appendix, as well as their equal part
versions for comparison. Results (running time) are presented in
Table~\ref{tab:times3}. In accordance with estimations of
Section~\ref{subsec:experimenttheor}, here we observe a clear
improvement only for pattern length $15$ and not for longer
patterns.

\begin{table}[!tb]
\caption{Total time (in sec) of search for $10^5$ patterns
in \emph{hr14}, up to 2 mismatches.
2nd column contains time obtained on partition into three equal-size parts. 
The 3rd (respectively 4th and 5th) column shows the running time
respectively for the $3$-unequal-parts, 
$4$-equal-parts and
$4$-unequal-parts searches, together with their ratio (\%) to the
corresponding $3$-equal-parts value.
\label{tab:times2}}
\centering
\begin{tabular}{|c|c|c|c|c|c|c|}
\hline
$m$ & 3 equal & \multicolumn{2}{|c|}{3 unequal} & 4 equal & \multicolumn{2}{|c|}{4 unequal} \\
\hline
15 & 24.8 & 25.4 (102\%) & 6,6,3 & 25.3 (102\%) & 25.3 (102\%) & 3,5,1,6 \\
24 & 5.5 & 4.2 (76\%)  & 10,7,7 & 5.2 (95\%) & 4.0 (73\%) & 7,4,4,9 \\
33 & 1.73  & 1.45 (84\%)   & 13,10,10 & 2.07 (120\%)  & 1.25 (72\%) & 11,5,6,11 \\
42 & 0.71  & 0.71 (100\%)  & 14,14,14 & 1.24 (175\%)  & 0.82 (115\%) & 14,6,8,14 \\
\hline
\end{tabular}
\end{table}


\begin{table}[!tb]
\caption{Total time (in sec) of search for $10^5$ patterns
  in \emph{hr14}, up to 3 mismatches.
\label{tab:times3}}
\centering
\begin{tabular}{|c|c|c|c|c|}
\hline
m & 4 equal & 5 equal & \multicolumn{2}{|c|}{5 unequal} \\
\hline
15 & 241 & 211 (86\%) & 206 (85\%) & 2,3,5,1,4 \\
24 & 19.7 & 26.7 (136\%) & 19.6 (99\%) & 2,9,3,1,9 \\
33 & 4.3 & 6.9 (160\%) & 4.7 (109\%) & 6,9,6,1,11 \\
42 & 1.85 & 2.52 (136\%) & 2.05 (111\%) & 10,10,9,1,12 \\
51 & 1.07 & 1.57 (147\%) & 1.06 (99\%) & 12,13,12,1,13 \\
\hline
\end{tabular}
\end{table}

\subsubsection{Edit distance}

In the case of edit distance, along with the search schemes for 2 and
3 errors from the previous section, we also implemented search schemes
for 4 errors (see Appendix). 
Results are shown in Table~\ref{tab:times2edit} 
(2 errors),
Table~\ref{tab:times3edit} (3 errors) and
Table~\ref{tab:times4edit} (4 errors). 
%

\begin{table}[!tb]
\caption{Total time (in sec) of search for $10^4$ patterns
in \emph{hr14}, up to 2 errors (edit distance).
\label{tab:times2edit}}
\centering
\begin{tabular}{|c|c|c|c|c|c|c|}
\hline
$m$ & 3 equal & \multicolumn{2}{|c|}{3 unequal} & 4 equal & \multicolumn{2}{|c|}{4 unequal} \\
\hline
15 & 11.5 & 11.4 (99\%) & 6,6,3 & 10.9 (95\%) & 11.1 (97\%) & 3,5,1,6 \\
24 & 2.1 & 1.3 (62\%)  & 11,5,8 & 1.5 (71\%) & 1.0 (48\%) & 7,4,4,9 \\
33 & 0.34  & 0.22 (65\%)   & 13,10,10 & 0.35 (103\%)  & 0.19 (56\%) & 11,5,6,11 \\
42 & 0.08  & 0.08 (100\%)  & 14,14,14 & 0.18 (225\%)  & 0.08 (100\%) & 14,6,8,14 \\
\hline
\end{tabular}
\end{table}


\begin{table}[!tb]
\caption{Total time (in sec) of search for $10^4$ patterns
  in \emph{hr14}, up to 3 errors (edit distance).
\label{tab:times3edit}}
\centering
\begin{tabular}{|c|c|c|c|c|}
\hline
m & 4 equal & 5 equal & \multicolumn{2}{|c|}{5 unequal} \\
\hline
15 & 233 & 174 (75\%) & 168 (72\%) & 2,2,6,1,4 \\
24 & 13.5 & 13.2 (98\%) & 10.8 (80\%) & 3,8,3,1,9 \\
33 & 0.74 & 1.81 (245\%) & 1.07 (145\%) & 5,10,5,1,12 \\
42 & 0.28 & 0.45 (161\%) & 0.37 (132\%) & 9,10,9,1,13 \\
51 & 0.13 & 0.24 (185\%) & 0.14 (108\%) & 12,12,12,1,14 \\
\hline
\end{tabular}
\end{table}

\begin{table}[!tb]
\caption{Total time (in sec) of search for $10^4$ patterns
in \emph{hr14}, up to 4 errors (edit distance).
\label{tab:times4edit}}
\centering
\begin{tabular}{|c|c|c|c|c|c|c|}
\hline
$m$ & 5 equal & \multicolumn{2}{|c|}{5 unequal} & 6 equal & \multicolumn{2}{|c|}{6 unequal} \\
\hline
15 & 4212 & 3222 (76\%)  & 3,1,8,1,2 & 4028 (96\%) & 3401 (81\%) & 2,2,1,7,1,2 \\
24 & 145 & 133 (92\%)  & 7,3,5,1,8 & 131 (90\%) & 113 (78\%) & 2,7,3,4,5,3 \\
33 & 6.5  & 5.8 (89\%)   & 8,7,5,8,5 & 6.6 (102\%)  & 5.1 (78\%) & 4,8,6,3,5,7 \\
42 & 1.66  & 1.16 (70\%)  & 12,8,7,8,7 & 1.51 (91\%)  & 1.17 (70\%) & 7,8,8,5,2,12 \\
51 & 0.60  & 0.49 (82\%)  & 13,11,9,9,9 & 0.74 (123\%)  & 0.54 (90\%) & 9,10,9,9,1,13 \\
60 & 0.28  & 0.24 (86\%)  & 14,13,11,11,11 & 0.44 (157\%)  & 0.28 (117\%) & 11,12,11,11,1,14 \\
\hline
\end{tabular}
\end{table}

For $2$ errors, we observe up to two-fold speed-up for pattern lengths
$15$, $24$ and $33$. For the case of $3$
errors, the improvement is achieved for pattern lengths $15$ and $24$
(respectively 28\% and 20\%). Finally, for $4$ errors, we obtain a significant speed-up (18\%
to 30\%) for pattern lengths between $15$ and $51$.

\subsubsection{Experiments on simulated 
genomic
  reads}\label{subsubsec:experimentgenomic}

\begin{table}[!tb]
\caption{Total time (in sec) 
of search for $10^5$ reads in \emph{hr14}, up to 4 errors.
First row corresponds to read set with constant error rate $0.03$.
Second row corresponds to read set with error rate increasing from $0.0$ to $0.03$.
\label{tab:expgenomic}}
\centering
\begin{tabular}{|c|c|c|c|c|c|c|}
\hline
$m$ & 5 equal & 6 equal & \multicolumn{2}{|c|}{6 unequal} \\
\hline
100 & 247  & 250 (101\%)  & 283 (115\%) & 20,20,20,19,1,20 \\
100 & 415  & 367 (88\%)  & 350 (84\%) & 20,20,20,19,1,20 \\
\hline
\end{tabular}
\end{table}

Experiments of Section~\ref{subsec:experimentpract} have been made with
random patterns. In order to make experiments closer to 
the practical bioinformatic setting 
occurring in mapping genomic reads to their
reference sequence, we also experimented with patterns simulating
reads issued from genomic sequencers. For that, 
we generated realistic single-end reads of length $100$ (typical length of {\sc
  Illumina} reads) from \emph{hr14} 
using {\sc dwgsim} read simulator
(\url{https://github.com/nh13/DWGSIM}). Two sets of reads were
generated using two different error rate
values (parameter \texttt{-e} of {\sc dwgsim}): \texttt{0.03} for the
first dataset and \texttt{0.0-0.03} for the second one. This means
that in the first set, error probability is uniform over the read
length, while in the second set, this probability gradually increases from $0$
to $0.03$ towards the right end of the read. The latter simulates
the real-life situation occurring with current sequencing
technologies including {\sc Illumina}. 

The results are shown in Table~\ref{tab:expgenomic}. As expected, due
to a large pattern length, our schemes did not produce a speed-up for
the case of constant error rate. 
Interestingly however,
for the case of non-uniform distribution of errors, our schemes showed
a clear advantage. 
This illustrates another possible benefit of our techniques: they are
better adapted to a search for patterns with non-uniform
distribution of errors, which often occurs in practical situations
such as mapping genomic reads. 

\section{Conclusions} 
This paper can be seen as the first step towards an automated design
of efficient search schemes for approximate string matching, based on
bidirectional indexes. More research has to be done in order to allow
an automated design of optimal search schemes. 
It would be very interesting to study an approach when a search scheme
is designed simultaneously with the partition, rather than
independently as it was done in our work. 

We expect that search schemes similar to those studied in this paper
can be applied to hybrid approaches to approximate matching (see Introduction), as well
as possibly to other search strategies. 



\paragraph{Acknowledgements.}
GK has been supported by the ABS2NGS grant of the French government
(program \emph{Investissement d'Avenir}) as well as by a EU Marie-Curie
Intra-European Fellowship for Carrier Development.
KS has been supported by the \emph{co-tutelle} PhD fellowship of the French
government.
DT has been supported by ISF grant 981/11.
\bibliographystyle{plain}
\bibliography{string-index,string-read-alignment}

\newpage
\section*{Appendix}
The following search schemes were used in experiments described in
Section~\ref{sec:experiments}. 

\noindent
For 2 mismatches or errors:
\begin{enumerate}
\item Slightly modified scheme $\schemelam$.
The searches are:
$S_f = (123,000,022)$, $S_b = (321,000,012)$,
and $S'_{bd} = (213,001,012)$.
Note that the $\pi$-string of $S'_{bd}$ is $213$ and not $231$ as in $S_{bd}$.
While  $S_{bd}$ and $S'_{bd}$ have the same efficiency for equal-size
partitions, this in not the case for unequally sized parts.
\item $4$-part scheme with searches
$(1234,0000,0112)$, $(4321,0000,0122)$, \linebreak[4]$(2341,0001,0012)$, and
$(1234,0002,0022)$.
\end{enumerate}

\vspace{5mm}
\noindent
For 3 mismatches or errors:
\begin{enumerate}
\item $4$-part scheme with searches 
$(1234,0000,0133)$, $(2134,0011,0133)$, \linebreak[4]$(3421,0000,0133)$, and
$(4321,0011,0133)$.
\item $5$-part scheme with searches
$(12345,00000,01233)$, $(23451,00000,01223)$, \linebreak[4]$(34521,00001,01133)$, and
$(45321,00012,00333)$.
\end{enumerate}

\vspace{5mm}
\noindent
For 4 mismatches or errors:
\begin{enumerate}
\item $5$-part scheme with searches
$(12345,00000,02244)$, $(54321,00000,01344)$, 
\linebreak[4]$(21345,00133,01334)$, $(12345,00133,01334)$, 
$(43521,00011,01244)$, \linebreak[4]$(32145,00013,01244)$, 
$(21345,00124,01244)$ and $(12345,00034,00444)$.

\item $6$-part scheme with searches
$(123456,00000,012344)$, $(234561,00000,012344)$, 
$(654321,000001,012244)$, $(456321,000012,011344)$, 
$(345621,000023,011244)$, $(564321,000133,003344)$, 
$(123456,000333,003344)$, $(123456,000044,002444)$,
$(342156,000124,002244)$ and $(564321,000044,001444)$.
\end{enumerate}

\end{document}